\documentclass[smallabstract,smallcaptions]{dccpaper}
\usepackage{amsfonts,amsmath, amssymb, amsthm, graphicx, cite}
\usepackage{bm}
\usepackage{algorithm,algorithmic}
\usepackage[latin1]{inputenc}
\usepackage{bbm}
\usepackage{amsmath}
\usepackage{amssymb}
\usepackage{graphics}
\usepackage{graphicx}
\usepackage{epstopdf}
\usepackage{latexsym}
\usepackage{amsfonts}
\usepackage{color}
\usepackage{url}
\usepackage{amsmath}
\usepackage{multirow}
\usepackage{tikz}
\usetikzlibrary{positioning,chains,fit,shapes,calc} 
\usetikzlibrary{shapes,arrows}
\usetikzlibrary{arrows}
\usepackage{color}
\usepackage{scalefnt}
\usepackage{textcomp}
\usepackage{booktabs}
\usepackage{enumitem}

\newtheorem{theorem}{Theorem}

\newtheorem{corollary}{Corollary}

\newtheorem{example}{Example}

\newtheorem*{case*}{Case}
\theoremstyle{definition}
\newtheorem{definition}{Definition}
\newtheorem{condition}{Condition}
\newcommand\norm[1]{\left\lVert#1\right\rVert}
\newcommand\f[1]{\footnotesize{#1}}
\newcommand\s[1]{\scriptsize{#1}}
\newcommand{\cmmnt}[1]{}
\usepackage{subcaption}

\definecolor{myblue}{RGB}{0,0,100}
\definecolor{mygreen}{RGB}{0,0,100}
\definecolor{myyellow}{RGB}{0,100,100}
\definecolor{my1}{RGB}{100,0,100}
\definecolor{myred}{RGB}{100,0,0}
\definecolor{my2}{RGB}{0,100,0}

\usepackage{epsfig}
\usepackage{amsmath}
\usepackage{amssymb}
\usepackage{color}
\usepackage{url}

\newcommand\blfootnote[1]{%
  \begingroup
  \renewcommand\thefootnote{}\footnote{#1}%
  \addtocounter{footnote}{-1}%
  \endgroup
}

\newlength{\figurewidth}
\newlength{\smallfigurewidth}

\begin{document}

\title
{\large
\textbf{Functional Epsilon Entropy}
}
\author{%
Sourya Basu$^{\ast}$, Daewon Seo$^{\dag}$, and Lav R. Varshney$^{\ast}$\\
{\small\begin{minipage}{\linewidth}\begin{center}
\begin{tabular}{ccc}
$^{\ast}$Coordinated Science Laboratory, && $^{\dag}$Department of Electrical Engineering\\
Department of Electrical and&& University of Southern California\\
Computer Engineering && \{\url{daewonse}\}@usc.edu\\
University of Illinois at Urbana-Champaign && {}\\
\{\url{sourya, varshney}\}@illinois.edu && {}
\end{tabular}
\end{center}\end{minipage}}
}
\maketitle
\Section{Abstract}
We consider the problem of coding for computing with maximal distortion, where the sender communicates with a receiver, which has its own private data and wants to compute a function of their combined data with some fidelity constraint known to both agents. We show that the minimum rate for this problem is equal to the conditional entropy of a hypergraph and design practical codes for the problem. Further, the minimum rate of this problem may be a discontinuous function of the fidelity constraint. We also consider the case when the exact function is not known to the sender, but some approximate function or a class to which the function belongs is known and provide efficient achievable schemes.

\section{Introduction}\label{sec: introduction}
Consider the problem illustrated in Fig.~\ref{fig: Coding for computing with side information.} where the encoder observes $X \in \mathcal{X}$ and decoder observes $Y \in \mathcal{Y}$, and both the encoder and decoder want the decoder to compute $f(X,Y) \in \mathcal{Z}$ with a fidelity criterion for a given function $f : \mathcal{X} \times \mathcal{Y} \mapsto \mathcal{Z}$ known to both encoder and decoder, where $\mathcal{X}, \mathcal{Y}, \mathcal{Z}$ are all finite sets and $\mathcal{Z} \subset \mathbb{R}^d$ for some finite natural number $d$. The objective is to find the minimum number of bits the encoder must send such that the decoder can compute $f(X,Y)$ with a fidelity criterion $\epsilon \in [0,\infty)$, i.e.\ if $\widehat{f(X,Y)}$ is the estimate of $f(X,Y)$ obtained by the decoder, then the following should hold \blfootnote{This work was funded in part by the IBM-Illinois Center for Cognitive Computing Systems Research (C3SR), a research collaboration as part of the IBM AI Horizons Network; and in part by grant number 2018-182794 from the Chan Zuckerberg Initiative DAF, an advised fund of Silicon Valley Community Foundation.}
\begin{align}
\norm{f(X,Y) - \widehat{f(X,Y)}} \leq \epsilon,
\end{align}
where $\norm{z_1 - z_2}$ is the Euclidean distance between $z_1,z_2 \in \mathcal{Z}$.
We assume the function $f$ is to be computed for $N$ independent instances of $(X,Y)$ for large $N$.
\tikzstyle{decision} = [diamond, draw, 
    text width=4.5em, text badly centered, inner sep=0pt]
\tikzstyle{block} = [rectangle, draw,
    text width=5em, text centered, minimum height=2 em, node distance = 1.5 cm]
\tikzstyle{block_e} = [rectangle,
    text width=12em, text centered, minimum height=2 em, node distance = 0 cm]
\tikzstyle{block_L} = [rectangle, draw,
    text width=5em, text centered, minimum height=8em, node distance = 4.2 cm]
\tikzstyle{line} = [draw, -latex']
\tikzstyle{cloud} = [draw, ellipse, node distance=4cm,
    minimum height=2em]
\tikzstyle{rel line to}= [to path={|- +(\tikztotarget) \tikztonodes}]
\tikzstyle{virtual} = [coordinate]

 \begin{figure}
  \tikzstyle{c} = [rectangle, draw, dashed]    
\centering
 
\begin{tikzpicture}[auto]

    \node [virtual] (message_1) {};
    \node [virtual, below = of message_1] (message_v) {};
    \node [block, right = of message_1] (encoder_1) {Enc 1};
	\node [virtual, right = of message_v] (encoder_v) {};    
    \node [block, right = 1.5cm of encoder_1] (decoder) {Decoder};
    \node [virtual, below = of decoder] (message_2) {};
    \node [block_e, right = of decoder] (fidelity) {\vspace{10mm}\\$\norm{f(X,Y) - \widehat{f(X,Y)}} \leq \epsilon$};
    

	\draw[->] (message_1.east) -- node[] {$X$ }+(1.5cm,0pt);
	\draw[->] (message_2.west) -- node[] {$Y$ }+(0cm,1cm);
	\draw[->] (encoder_1.east) -- node[] {$m$} +(1.5cm,0pt);
	\draw[->] (decoder.east) -- node[] {$\widehat{f(X,Y)}$} +(3cm,0pt);
\end{tikzpicture}

\caption{Coding for computing with side information.}
\label{fig: Coding for computing with side information.}
\end{figure}
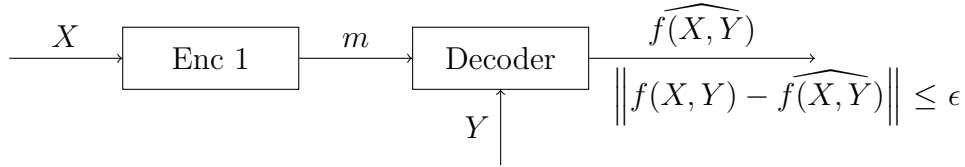

Orlitsky and Roche gave a single-letter characterization of the problem for $\epsilon = 0$ \cite{OrlitskyR2001}. Function computation with $\epsilon > 0$ has been considered in \cite{DoshiSMJ2010,FeiziM2014} which provide an efficient graph-based encoding scheme using a generalized version of the characteristic graph in \cite{OrlitskyR2001} called the $D$-characteristic graph. The coding scheme provided there is suboptimal since the construction of the graph is based on only pairwise comparison of function values. Drawing on better geometric methods of comparison, we introduce a novel generalization of the characteristic graph from \cite{OrlitskyR2001} that yields a hypergraph-based encoding scheme that is indeed optimal. Thus, we provide an alternate (but fully equivalent) description of the rate-distortion function, which further inspires practically implementable codes. We refer to the optimal rate for this problem as \emph{functional $\epsilon$-entropy} since this rate reduces to \emph{$\epsilon$-entropy} of $f$ defined in \cite{PosnerR1971} when there is no side information $Y$ and the hyperedges of the constructed characteristic hypergraph partition the support set $\mathcal{X}$.  Unlike traditional rate-distortion problems where the rate-distortion function $R(\epsilon)$ is a continuous function of $\epsilon$ (the constraint on the expected distortion), in this case we have a rate function that is discontinuous function of $\epsilon$ and the points of discontinuity can be determined from the characteristic hypergraph for different values of $\epsilon$. 

We also show that some of the assumptions in \cite{OrlitskyR2001} that lead to simple coding schemes for $\epsilon = 0$ might not imply the same when $\epsilon > 0$. Although \cite{DoshiSMJ2010, FeiziM2014} show modular schemes, i.e.\ graph-based quantization followed by source coding, are optimal for the problem with $\epsilon = 0$ under some assumptions, the solution provided is NP-hard and an approximate solution is used for coding. Further, the assumptions in \cite{DoshiSMJ2010, FeiziM2014, OrlitskyR2001} that lead to optimal modular schemes only involve the source distributions; one can weaken assumptions that imply optimal modular schemes by considering both the distribution of the source and the function. Although there does not seem to be a simple and elegant dichotomy of functions and sources analogous to Han and Kobayashi's dichotomy of functions \cite{HanK1987} under which modular schemes are optimal, we provide a simple and general class of function-source pairs for which modular schemes are optimal. 

After providing the rate-distortion function and an optimal hypergraph-based encoding scheme, we provide algorithms and conditions under which practical coding schemes using \emph{randomized quantization} and polar codes \cite{Arikan2009}  are optimal.

Sec.~\ref{sec_merged: notations and preliminaries} describes preliminary results and the problem model. Sec.~\ref{sec: graph} gives the main result of this paper, the coding theorem and its equivalence to the conditional entropy of a hypergraph. Sec.~\ref{sec: practical} develops practical coding schemes for the problem. Sec.~\ref{sec: properties} shows that functional $\epsilon$-entropy may be discontinuous in $\epsilon$ and Sec.~\ref{sec: conclusion} concludes the paper.

\section{Preliminaries and problem setting}\label{sec_merged: notations and preliminaries}
In this section, first we provide some definitions and discuss some preliminary results; then we formally define the problem.

\subsection{Preliminaries and notations}\label{sec: notations and preliminaries}
For a set of $k$ points $z^k = \{z_1, \ldots, z_k\}$, $z_i \in \mathcal{Z} \subset \mathbb{R}^d$ for some finite $d$, the smallest circle (sphere) enclosing these $k$ points is called the smallest enclosing circle of $z^k$ \cite{Chrystal1885}. Note that the computational complexity of finding the smallest enclosing circle for a set of points is linear in the number of points \cite{Megiddo1983}.

\begin{definition}
A function $f: \mathcal{X} \mapsto \mathcal{Z}$ is $L$-Lipschitz continuous if for all $x_1,x_2 \in \mathcal{X}$, $\norm{f(x_1) - f(x_2)} \leq L \norm{x_1 - x_2}$ for some $L>0$ where $\norm{\cdot}$ is the Euclidean norm.
\end{definition}
\begin{definition}
A function $g: \mathcal{X} \mapsto \hat{\mathcal{Z}}$, $\hat{\mathcal{Z}} \subset \mathbb{R}^d$ is a $\delta$-approximation to a function $f: \mathcal{X} \mapsto \mathcal{Z}$ if for every $x \in \mathcal{X}$, $\norm{f(x) - g(x)} \leq \delta$ where $\norm{\cdot}$ is the Euclidean norm.
\label{def: delta}
\end{definition}

A hypergraph $G$ is a pair $G = (\mathcal{X},E)$ where $\mathcal{X}$ is the set of vertices of $G$ and $E \subseteq \mathcal{P}(\mathcal{X}) \setminus \emptyset$ is the set of hyperedges of $G$, where $\mathcal{P}(\mathcal{X})$ is the powerset of $\mathcal{X}$ \cite{Bretto2013}. 

\begin{definition}
A hyperedge $E$ is called a \emph{maximal hyperedge} if $E$ is not a proper subset of any other hyperedge in the hypergraph $G$.
\end{definition}

Let $(W,X) \sim P_{WX}$ be i.i.d. random variables with $(W,X) \in \mathcal{W}\times \mathcal{X}$ and suppose $\mathcal{W} = \{0,1\}$ for simplicity. Then, source Bhattacharyya parameter \cite{Arikan2010} $Z(W|X)$ for the source $(W,X)$ is defined as $Z(W|X) = 2 \sum_{x} P_{X}(x) \sqrt{P_{W|X}(0|x) P_{W|X}(1|x)}$.

\begin{theorem}[\cite{HondaY2013}]
For any $\beta < \tfrac{1}{2}$, i.i.d.\ random variables $(W,X)$ and $U_1^N = W_1^N G_N$,
\begin{align*}
\lim_{N \to \infty} \frac{1}{N} \left \vert \{  i: Z(U_i|U_1^{i-1},X_1^N) \leq 2^{-N^{\beta}} \text{  and  } Z(U_i|U_1^{i-1}) \geq 1 - 2^{-N^{\beta}} \} \right \vert & = I(W;X),\\
\lim_{N \to \infty} \frac{1}{N} \left \vert \{  i: Z(U_i|U_1^{i-1},X_1^N) \geq 1 - 2^{-N^{\beta}} \text{  or  } Z(U_i|U_1^{i-1}) \leq 2^{-N^{\beta}} \} \right \vert & = 1 - I(W;X),
\end{align*}
where $G_N$ is the generator matrix for polar codes. 
\label{thm: polar_codes}
\end{theorem}
We will use polar coding to build practical coding techniques for our problem. Thm.~\ref{thm: polar_codes} can be extended to any finite set $\mathcal{W}$ using results from \cite{SasogluTA2009}.

\subsection{Problem setting}\label{sec: problem_Setup}
Let $(X_i,Y_i) \sim P_{X,Y}$ be $N$ i.i.d.\ random variables, where $X_i \in \mathcal{X}$, $Y_i \in \mathcal{Y}$. The encoder in Fig.~\ref{fig: Coding for computing with side information.} observes $\{X_i\}_{i=1}^N$, the decoder observes $\{Y_i\}_{i=1}^N$, and both the encoder and decoder want the decoder to reconstruct $\{f(X_i,Y_i)\}_{i=1}^{N}$ as $\{\hat{Z}\}_{i=1}^{N}$ such that $\frac{1}{N}\sum_{i=1}^{N}\Pr \left[\norm{\hat{Z}_i - f(X_i,Y_i)} > \epsilon\right] \to 0$ as $N \to \infty$ for some fixed fidelity constraint $\epsilon > 0$.
For any $R>0$, we define a $(2^{NR},N, \epsilon)$ code for any fixed function $f: \mathcal{X}\times \mathcal{Y} \to \mathcal{Z}$ as an encoding function $g_e: \mathcal{X}^N \mapsto \{1,\ldots, 2^{NR}\}$ and a decoding function $g_d: \{1,\ldots, 2^{NR}\} \times \mathcal{Y}^N \mapsto \hat{\mathcal{Z}}^N$ where $\hat{\mathcal{Z}} \subset \mathbb{R}^d$ is the reconstruction set. The probability of error is 
\[
P_{\epsilon}^{avg}(\hat{Z}, X, Y) = \frac{1}{N}\sum_{i=1}^{N}\Pr \left[\norm{\hat{Z}_i - f(X_i,Y_i)} > \epsilon\right],
\]
that is, $P_{\epsilon}^{avg}(\hat{Z}, X, Y)$ is the average symbol-error probability. A rate $R$ is achievable if there exists a sequence of $(2^{NR},N, \epsilon)$ codes such that $P_{\epsilon}^{avg}(\hat{Z}, X, Y) \to 0$ as $N \to \infty$. The goal is to find the minimum achievable value of $R$ and design practical codes that attain it. 

\section{Coding theorem and hypergraph-based coding scheme}\label{sec: graph}

In this section we first provide the rate-distortion function for the problem described in Sec.~\ref{sec_merged: notations and preliminaries}, and then provide a hypergraph-based coding scheme that achieves the rate-distortion function.
\begin{theorem}
Let $(X,Y) \sim P_{XY}$. The encoder and decoder observe $X$ and $Y$ respectively. The decoder estimates the function $f(X,Y)$ as $\widehat{f(X,Y)}$ such that for some fixed $\epsilon > 0$, $P_{\epsilon}^{avg}(\widehat{f(X,Y)}, X, Y) \to 0$ as $N \to \infty$. Then the minimum rate required by the encoder is
\begin{align}
R(\epsilon) = \min_{U-X-Y } I(U;X|Y) \label{eqn: R_Delta}
\end{align}
such that there exists a function $g$ with $\mathbb{E}d_{\epsilon}(X,Y,g(U,Y)) \leq 0$, where the distortion function $d_{\epsilon}$ is defined as $d_{\epsilon}(x,y,z) = \mathbbm{1} \{ \norm{z - f(x,y)} > \epsilon \}.$
\label{thm: I:coding theorem}
\end{theorem}
The proof to this theorem is direct and follows from \cite[Eq.\ (4)]{OrlitskyR2001} when the distortion is set to zero under the $d_{\epsilon}$ distortion function.

We define the \emph{$\epsilon$-characteristic hypergraph}, $G_{f, \epsilon;X|Y}$, of a random variable $X$ with respect to another possibly correlated random variable $Y$, a function $f$, and a fidelity constraint $\epsilon$.

\begin{definition}
The vertex set of \emph{$\epsilon$-characteristic hypergraph}, $G_{f,\epsilon;X|Y}$, is $\mathcal{X}$. For any non-empty subset $S \subseteq \mathcal{X}$ and $y \in \mathcal{Y}$, let $S_y = \{x:x\in S \hspace{2mm}\text{and} \hspace{2mm} p(x,y) > 0\}$. Then $S$ is a hyperedge in $G_{f,\epsilon;X|Y}$ if and only if the radius of the smallest enclosing circle containing the set of points $\{f(x,y): x \in S_y\}$ is less than or equal to $\epsilon$ for all $y \in \mathcal{Y}$.
\label{def: graph}
\end{definition}
Note that for $\epsilon = 0$ the hypergraph in Def.~\ref{def: graph} reduces to the characteristic graph defined in \cite{OrlitskyR2001, Witsenhausen1975} with hyperedges replaced by independent sets. Now we define the hypergraph entropy of a characteristic hypergraph $G_{f,\epsilon;X|Y}$. Let $\Gamma(G_{f,\epsilon;X|Y})$ be the set of hyperedges of $G_{f,\epsilon;X|Y}$. When it is clear from context, we will denote $G_{f,\epsilon;X|Y}$ by $G_{\epsilon}$, and $G_{0}$ is simply written as $G$. We define the functional $\epsilon$-entropy, which is a generalization of $\epsilon$-entropy proposed by \cite{PosnerR1971}.
\begin{definition}
The \emph{functional $\epsilon$-entropy}, $H_{G_{\epsilon}}(X|Y)$, is defined as
\begin{align}
H_{G_{\epsilon}}(X|Y) = \min_{\substack{W-X-Y \\ X \in W \in \Gamma (G_{\epsilon})}} I(W;X|Y),
\label{eqn: graph_entropy}
\end{align}
where $X$ induces a probability distribution over the vertices of the hypergraph $G_{\epsilon}$. The random variable $W$ is obtained by defining transition probabilities $p(w|x)$ over all hyperedges $w$ that contain $x$, i.e.\ $p(w|x) \geq 0$ for all $x \in w \in \Gamma(G_{\epsilon})$ and  $\sum_{w \ni x}p(w|x) = 1$.
\end{definition}
Note that the minimization over $\Gamma (G_{\epsilon})$ can be restricted to $\Gamma_m (G_{\epsilon})$ by the data processing inequality, where $\Gamma_m (G_{\epsilon})$ is the set of maximal hyperedges. Now we show that the optimal rate $R(\epsilon) = H_{G_{\epsilon}}(X|Y)$.
\begin{theorem}
Let $R({\epsilon})$ and $H_{G_{\epsilon}}(X|Y)$ be as defined in \eqref{eqn: R_Delta} and \eqref{eqn: graph_entropy} respectively for some $\epsilon \geq 0$, then $R({\epsilon}) = H_{G_{\epsilon}}(X|Y). $
\label{thm: 2Delta}
\end{theorem}

\begin{proof}
From the definitions of $R(\epsilon)$ and $H_{G_{\epsilon}}(X)$ we need to prove that 
\begin{align*}
\min_{\substack{U-X-Y \\ \exists g: \hspace{0.5em}\mathbb{E}d_{\epsilon}(X,Y,g(U,Y)) \hspace{0.1em} \leq \hspace{0.1em}0}} I(U;X|Y) = \min_{\substack{W-X-Y \\ X \in W \in \Gamma (G_{\epsilon})}} I(W;X|Y).
\end{align*}
First we show the left side is less than or equal to the right side.
If $X \in W \in \Gamma (G_{\epsilon})$, then we can find a (partial) function $g$ over $\Gamma (G_{\epsilon}) \times \mathcal{Y}$ such that $\norm{g(w,y) -  f(x,y)} \leq \epsilon$ whenever $p(w,x,y) > 0$, and thus implying $\mathbb{E}d_{\epsilon}(X,Y,g(W,Y)) = 0$.

Let $w \in \Gamma (G_{\epsilon})$ and $y \in \mathcal{Y}$. If $p(x,y) = 0$ for all $x \in w$, then we can leave $g$ undefined since it will not affect our expected distortion. Otherwise form the set $w_y$ which consists of all $x \in w$ such that $p(x,y) > 0$ and define $g(w,y)$ as the center of the smallest enclosing circle of the set $\{f(x,y): x \in w_y\}$. Then, by Def.~\ref{def: graph}, $\norm{ f(x,y) - g(w,y)} \leq \epsilon$ for all $x \in w$ such that $p(x,y) > 0$. Hence, whenever $p(w,x,y) > 0$, we have $d_{\epsilon}(x,y,g(w,y)) = 0$. This shows that $\mathbb{E}d_{\epsilon}(X,Y,g(W,Y)) = 0$, hence the left side is less than or equal to the right side. 

Next we show that the right hand side is less than or equal to the left hand side, completing the proof.
Suppose $U-X-Y$ and there exists a $g$ such that $\mathbb{E}d_{\epsilon}(X,Y,g(U,Y)) \leq 0$. We define $W$ such that $X \in W \in \Gamma (G_{\epsilon})$ and show that $I(W;X|Y) \leq I(U;X|Y)$ for this definition. Let $p(u,x,y)$ be the probability distribution underlying $(U,X,Y)$. Set
\begin{equation}
w(u) = \{x: p(u,x) > 0\}
\label{eqn: w_def}
\end{equation}
and define the Markov chain $W-U-XY$ by
\begin{align*}
p(w|u,x,y) = 
\begin{cases}
1, & \text{if $w = w(u)$}\\
0, & \text{otherwise.}
\end{cases}
\end{align*}
We first show that $X \in W \in \Gamma (G_{\epsilon})$. If $p(w,x)>0$, this implies there is a $u$ such that $w = w(u)$ and $p(u,x) > 0$. Then, by \eqref{eqn: w_def} we have $x \in w = w(u)$. Thus, we have $x \in w$ whenever $p(w,x) > 0$. Next we show that whenever $p(w) > 0$, then the radius of the smallest enclosing circle of the set $\{f(x,y): x \in w \hspace{2mm} \text{and} \hspace{2mm} p(x,y) > 0\}$ is less than or equal to $\epsilon$ which further implies $W \in \Gamma (G_{\epsilon})$. From \eqref{eqn: w_def}, $p(w) > 0$ implies there exists a $u$ such that $w = w(u)$. Further, if $x \in w$ then $p(u,x) > 0$, and since $U-X-Y$ forms a Markov chain, it follows that whenever $p(x,y) > 0$ we have $p(u,x,y) > 0$. Note that $\mathbb{E}d_{\epsilon}(X,Y,g(U,Y)) = 0$, hence we must have $\norm{f(x,y) -  g(u,y)} \leq \epsilon$ whenever $p(u,x,y) > 0$. Thus, it follows that the circle centered at $g(u,y)$ of radius $\epsilon$ encloses all the points in the set $\{f(x,y): x \in w \hspace{2mm} \text{and} \hspace{2mm} p(x,y) > 0\}$. Hence, the smallest enclosing circle of the set $\{f(x,y): x \in w \hspace{2mm} \text{and} \hspace{2mm} p(x,y) > 0\}$ has radius less than or equal to $\epsilon$, which implies $w \in \Gamma (G_{\epsilon})$.

It remains to show that $W-X-Y$ forms a Markov chain and that $I(W;X|Y) \leq I(U;X|Y)$. The proof to this follows directly from the proof of \cite[Thm.~2]{OrlitskyR2001}.
\end{proof}

Note that setting $\epsilon = 0$, gives us the rate-distortion function in \cite{OrlitskyR2001}. 

\section{Towards practical coding scheme}\label{sec: practical}
In \cite{OrlitskyR2001} it was shown that if $p(x,y)>0$ for all $(x,y) \in \mathcal{X} \times \mathcal{Y}$, then every vertex $x$ in the characteristic hypergraph $G$ belongs to exactly one hyperedge of $G$. This property led to the design of optimal modular schemes in \cite{DoshiSMJ2010,FeiziM2014} under some assumptions introduced therein. However, those assumptions depended solely on the sources and not on the function considered. In this section, we introduce a function/source condition and show that for $\epsilon = 0$, this implies non-overlapping clustering of vertices leading to a modular scheme that can be implemented in $\mathcal{O}(N\log{N})$ time where $N$ is the blocklength. Then, by giving a counterexample we show that for $\epsilon >0$, $p(x,y)>0$ for all $(x,y) \in \mathcal{X} \times \mathcal{Y}$ does not imply non-overlapping clustering  of vertices in $G_{\epsilon}$ in contrast to the case $\epsilon = 0$ where this condition on the source implies non-overlapping clustering.

\begin{condition}
For any $y \in \mathcal{Y}$ and $x,x' \in \mathcal{X}$, if $f(x,y) \neq f(x',y)$, then either $p(x,y) = p(x',y) = 0$ or $p(x,y)>0, p(x',y)>0$.
\label{cond: 1}
\end{condition} 
Note that Cond.~\ref{cond: 1} encompasses the cases $p(x,y) > 0$ for all $(x,y) \in \mathcal{X}\times \mathcal{Y}$ and $p(x,y) = p(x) p(y)$ for all $(x,y) \in \mathcal{X}\times \mathcal{Y}$, where $\mathcal{X}$ and $\mathcal{Y}$ are the support set of $X$ and $Y$ respectively. The main idea is that whenever Cond.~\ref{cond: 1} holds, each $x \in \mathcal{X}$ belongs to a unique maximal hyperedge in $\Gamma(G)$. Hence quantization followed by entropy coding attains the optimal rate. Consider the following example that illustrates Cond.~\ref{cond: 1} and each vertex belongs to exactly one maximal hyperedge even though $p(x,y) = 0$ for some $(x,y) \in \mathcal{X}\times \mathcal{Y}$.

\begin{example}
Consider the random variables $(X,Y) \in \mathcal{X} \times \mathcal{Y}$ distributed as $P_{XY}$ in Fig.~\ref{fig: example_1}.\  a and let $f(x,y)$ in Fig.~\ref{fig: example_1}.\  b be the corresponding function, where $\mathcal{X} = \mathcal{Y} = \{1,2,3\}$. The hypergraph formed in this case is shown in Fig.~\ref{fig: example_1}.\ c.
\label{example: assumption}
\end{example}

\begin{figure}
\centering
 \begin{tabular}{lrrr}
	\toprule
	\f{$X\backslash Y$} & \f{1} & \f{2} & \f{3} \\
	\midrule
	\f{1} & \f{ $\tfrac{1}{7}$ } & \f{ $\tfrac{1}{7}$ }  & \f{0}\\ 
	\f{2} & \f{ $\tfrac{1}{7}$ }  & \f{ $\tfrac{1}{7}$ }  & \f{$\tfrac{1}{7}$}\\
	\f{3} & \f{ $\tfrac{1}{7}$ }  & \f{ $\tfrac{1}{7}$ }  & \f{0} \\
	\bottomrule
\end{tabular} \ 
\hspace{0.5cm}
\begin{tabular}{lrrr}
	\toprule
	\f{$X\backslash Y$} & \f{1} & \f{2} & \f{3} \\
	\midrule
	\f{1} & \f{1} & \f{1} & \f{1} \\
	\f{2} & \f{1} & \f{0} & \f{1}\\
	\f{3} & \f{1} & \f{0} & \f{1} \\
	\bottomrule
\end{tabular} \
\hspace{0.3cm}
 \includegraphics[width=60mm]{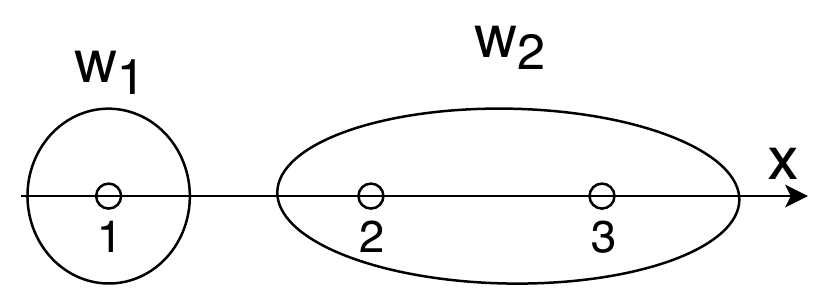} \\ \vspace{1mm}
\hspace{-1.5cm} (a) \hspace{3.5cm} (b) \hspace{4.5cm} (c)
\caption{(a) Probability distribution for $(X,Y)$. (b) Function $f(X,Y)$. (c) Corresponding hypergraph $G_0$.}
\label{fig: example_1}
\end{figure}
\subsection{Modular schemes}\label{I: practical_modular}
We show that whenever Cond.~\ref{cond: 1} holds, any $x \in \mathcal{X}$ belongs to exactly one maximal hyperedge $w$ in $\Gamma(G)$.

\begin{theorem}
If Cond.~\ref{cond: 1} holds, then for any $x \in \mathcal{X}$, if $x \in w_1, w_2$ for $w_1, w_2 \in \Gamma_m(G)$, then $w_1 = w_2$.
\label{thm: unique_clustering}
\end{theorem}
\begin{proof}
Without loss of generality, assume that $|w_1| \leq |w_2|$. We know from Def.~\ref{def: graph} that for any $w \in \Gamma(G)$, $x_1,x_2 \in w$ if and only if for all $y \in \mathcal{Y}$ either $f(x_1,y) = f(x_2,y)$, $p(x_1,y) = 0$, or $p(x_2,y) = 0$ holds. If $w_1$ is a singleton set then we are done since this implies $w_1 \subseteq w_2$ but since $w_1$, $w_2$ are maximal sets, we have $w_1 = w_2$. Now take the case when both $w_1, w_2$ are not singleton sets. Assume that $w_1 \neq w_2$, then there is $x', x'' \in \mathcal{X} \setminus x$ such that $x' \neq x''$, and $x' \in w_1, x' \notin w_2$, and $x'' \notin w_1, x'' \in w_2$.

For any $y \in \mathcal{Y}$, under Cond.~\ref{cond: 1} one of the following cases hold:
\begin{enumerate}
\item $f(x,y) = f(x',y) = f(x'',y)$.
\item If $f(x,y) \neq f(x',y)$ (or $f(x,y) \neq f(x'',y)$), then $p(x,y) = p(x',y) = 0$ (or $p(x,y) = p(x'',y) = 0$) by Cond.~\ref{cond: 1}.
\end{enumerate}

Hence, for all $y \in \mathcal{Y}$, we have $f(x',y) = f(x'',y)$ or $p(x',y) = 0$ or $p(x'',y) = 0$, which implies that $x'$ and $x''$ belong to the same hyperedge in $G$ by Def.~\ref{def: graph}. Thus, $x', x''$ belongs to the same maximal hyperedge which implies $w_1 \subseteq w_2$. But since $w_1$ and $w_2$ are maximal sets, it implies $w_1 = w_2$.
\end{proof}

Thm.~\ref{thm: unique_clustering} implies that whenever Cond.~\ref{cond: 1} holds, each $x \in \mathcal{X}$ belongs to exactly one maximal hyperedge in $\Gamma(G)$. Thus, hypergraph-based coding implies the following quantization$+$entropy coding scheme attains the optimal rate. Given any $x \in \mathcal{X}$, encode it using the unique hyperedge it belongs to and then use Slepian-Wolf coding to achieve the rate-distortion function. This scheme can be implemented in $\mathcal{O}(N\log{N})$ time since quantization can be performed in constant time and Slepian-Wolf coding can be implemented in $\mathcal{O}(N \log{N})$ time using polar codes \cite{Arikan2010} where $N$ is the blocklength.

Next consider the case when $\epsilon > 0$. Unlike for the case $\epsilon = 0$, when $\epsilon > 0$, even when $p(x,y) > 0$ for all $(x,y) \in \mathcal{X} \times \mathcal{Y}$ we might not have non-overlapping clustering, i.e.\ we can have vertices belonging to more than one maximal hyperedge.

\begin{figure}[H]
\begin{center}
\includegraphics{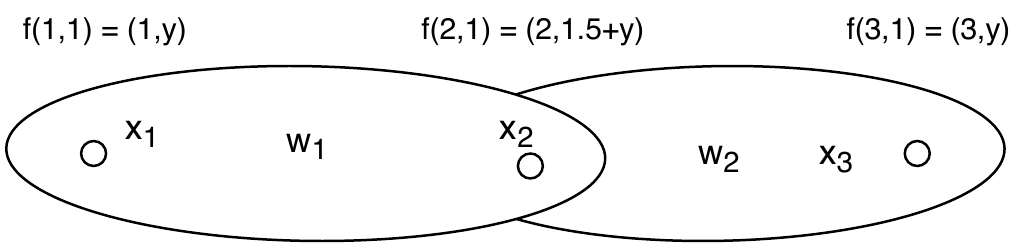}
\end{center}
\caption{The hypergraph consisting of vertex set $\mathcal{X} = \{x_1,x_2,x_3\}$ and hyperedges $w_1 = \{x_1, x_2\}$ and $w_2 = \{x_2,x_3\}$.}
\label{fig: counterexample_hypergraph.}
\end{figure}

\begin{example}
Let $X$ and $Y$ be independent uniform random variables defined on the support set $\mathcal{X} = \{1,2,3\}$ and $\mathcal{Y} = \{1,2\}$ respectively. Let $f: \mathcal{X} \times \mathcal{Y} \to \mathcal{Z}$, where $\mathcal{Z} \subset \mathbb{R}^2$, and $f$ be defined as $f(1,y) = (1,y)$, $f(2,y) = (2,1.5 + y)$, $f(3,y) = (3,y)$, and let $\epsilon = \frac{\sqrt{13}}{4}$. Then the characteristic hypergraph $G$ is as shown in Fig.~\ref{fig: counterexample_hypergraph.}. The hypergraph $G$ consists of three vertices $\{x_1,x_2,x_3\}$ and two maximal hyperedges $w_1$ and $w_2$. The smallest enclosing circle of the set of points $\{f(1,y), f(2,y)\}$ is centred at $(1.5,y + 0.75)$ and has a radius of $\epsilon = \frac{\sqrt{13}}{4}$. Hence, $w_1 = \{x_1,x_2\}$ forms a hyperedge of $G$. Similarly, $w_2 = \{x_2,x_3\}$ forms a hyperedge, but the smallest enclosing circle of $\{f(1,1), f(3,1)\}$ has a radius of $1$ which is greater than $\epsilon$ and hence $\{x_1,x_3\}$ does not form a hyperedge. Thus, we see that even though $p(x,y) > 0$ for all $(x,y) \in \mathcal{X}\times \mathcal{Y}$, $x_2$ belongs to two different maximal hyperedges $w_1$ and $w_2$.
\label{example: unique_clustering}
\end{example}
\subsection{Partially known functions}\label{II: unknown_functions}
Suppose there is no side information available at the decoder and the function $f$ is unknown to the encoder but it is known that $f$ is an $L$-Lipschitz continuous function. Then as a corollary of Thm.~\ref{thm: 2Delta}, we have the following result which may be of interest in several applications where the actual function is unknown or requires more computational resources than are available at the encoder. For instance if there is a good linear approximation to a computationally heavy function available to the encoder, the encoder might use the simpler function rather than the actual one.

\begin{corollary}
Let $f: \mathcal{X} \mapsto \mathcal{Z}$ be a $L$-Lipschitz continuous function. Then $R(\epsilon)$ can be upper-bounded as
$
R(\epsilon) \leq H_{G_{\epsilon / L}}(X),
$
where $G_{\epsilon / L}$ is constructed with respect to the random variable $X$ and the identity function and hence the upper-bound is achievable by the encoder even when $f$ is unknown. 
\label{cor: L-Lipschitz}
\end{corollary}
\begin{proof}
The proof follows from Thm.~\ref{thm: 2Delta} and the properties of $L$-Lipschitz continuous functions. The main idea is that if a set of $k$ points $\{x_1, x_2, \ldots, x_k\}$ has a smallest enclosing circle of radius $r$, then the set of points $\{f(x_1), f(x_2), \ldots, f(x_k)\}$ will have a smallest enclosing circle of radius $r' \leq rL$.
\end{proof}

Now, consider the case where the encoder cannot compute the exact function $f$ but computes $g$, which is an $\delta$-approximation to $f$ as defined in Def.~\ref{def: delta}.
\begin{corollary}
Let $g$ be a $\delta$-approximation to $f$. If the encoder only has access to $g$, then  for $\epsilon > 2 \delta$, $R(\epsilon)$ can be upper-bounded as
$
R(\epsilon) \leq H_{G_{(\epsilon - 2 \delta)}}(X),
$
where $G_{(\epsilon - 2 \delta)}$ is constructed with respect to the random variable $X$ and $g$. Moreover, this upper bound is achievable.
\label{cor: Delta - 2 epsilon}
\end{corollary}

\begin{proof}
Since $g$ is a $\delta$-approximation to $f$, if a set of points have a smallest enclosing circle of radius $r$ with respect to $g$, then the same set of points must have a smallest enclosing circle of radius less than or equal to $r + 2 \delta$ with respect to $f$. Hence, constructing a graph with fidelity constraint $\epsilon - 2 \delta$ and the function $g$ ensures that the maximal distortion with respect to $f$ is less than or equal to $\epsilon$.
\end{proof}

Although Ex.~\ref{example: unique_clustering} shows that for $\epsilon > 0$, there can be overlapping clustering even when $X$ and $Y$ are independent random variables, in Sec.~\ref{subsec: prob_quant_polar} we will show that in cases where there is overlapping clustering we can still use a \emph{randomized} form of quantization followed by polar coding to attain the optimal rate.

\subsection{Quantization and universal source coding  for lossless coding for computing}\label{subsec: quant_usc}

When $X$ is independent of $Y$, we have $p(x,y) = p(x)p(y) > 0$ for all $(x,y) \in \mathcal{X}\times \mathcal{Y}$, where $\mathcal{X}$ and $\mathcal{Y}$ are the support set of $X$ and $Y$ respectively. Hence it follows from \eqref{eqn: graph_entropy} that $H_{G}(X|Y) = H(q(X))$, where $q(X)$ is the quantized value of $X$ corresponding to the unique maximal hyperedge that $X$ belongs to. Moreover, note that the function $q$ depends only on the function $f$ and not on the probability mass function of $X$, which implies that for a fixed function $f$ and any $p(x,y)>0$ such that $X$ is independent of $Y$, there is a universal source coding scheme that attains the functional $\epsilon$-entropy in \eqref{eqn: graph_entropy} which is the minimum number of bits the encoder needs to send when $\epsilon = 0$. We illustrate this case using an example.  

\begin{example}
Let $f(X,Y) = \arg \min_{i \in \{1,2\}} g(i, X, Y)$ where $g(1,x,y) = x^2 + 3y$, $g(2,x,y) = x + 2y^2$ as illustrated in Fig.~\ref{fig: lzw}.a, $\arg \min$ takes the minimum value of $i$ in case of equality, and take $\epsilon = 0$. Let $X$ and $Y$ be independent random variables, $P_X$ be the probability mass function of $X \in \mathcal{X} = \{1,2,3,4\}$, and $P_Y$ be the probability mass function of $Y \in \mathcal{Y} = \{1,2\}$. Since $X$ and $Y$ are independent random variables, Thm.~\ref{thm: unique_clustering} implies the optimal rate can be obtained by using quantization followed by universal source coding. The quantization scheme for the function $f$ can be found to be forming two clusters in the characteristic hypergraph, i.e.\ any $x \in \{1,2\}$ maps to one cluster, while $x \in \{3,4\}$ maps to a different cluster. Each row of Fig.~\ref{fig: lzw}.b shows the probability mass function on $\mathcal{X}$, their corresponding entropy $H(X)$, the optimal rate of functional compression for the function $f$, $H_G(X) = H(q(X))$, where $q$ is the quantization function mapping $x \in \{1,2\}$ to one value and $x \in \{3,4\}$ to another, and the rate observed by using LZW algorithm \cite{ZivL1978} for blocklength $10^5$. 
\label{example: indepnedent_sources}
\end{example}
\begin{figure}
\centering
 \begin{tabular}{lrr}
	\toprule
	\f{$X$}$\backslash Y$ & \f 1 & \f 2\\
	\midrule
    \f{1} & \f 2 & \f 2\\ 
    \f 2 & \f 2 & \f 2\\ 
    \f 3 & \f 1 & \f 2\\ 
    \f 4 & \f 1 & \f 2\\ 
	\bottomrule
\end{tabular} \ 
\hspace{3cm}
\begin{tabular}{lrrr}
	\toprule
	 $\f P_X$ & $\f H(X)$ & $\f H_G(X)$ & \f LZW rate\\
	\midrule
	$[\frac{1}{15}, \frac{4}{15}, \frac{8}{15}, \frac{2}{15}]$ & \f 1.64 &\f  0.92 &\f 1.06\\
	$[\frac{2}{17}, \frac{1}{17}, \frac{8}{17}, \frac{6}{17}]$ & \f 1.65 &\f  0.67 & \f 0.80\\
	$[\frac{1}{3}, \frac{1}{3}, \frac{1}{3}, \frac{1}{6}]$ & \f 1.95 & \f 0.99 & \f 1.14\\
	$[\frac{1}{6}, \frac{1}{6}, \frac{5}{12}, \frac{1}{4}]$ & \f 1.88 & \f 0.92 & \f 1.06\\
	\bottomrule
\end{tabular} \\ \vspace{1mm}
\hspace{-2.3cm} (a) \hspace{8cm} (b)
\caption{(a) Function $f(X,Y)$. (b) $H(X)$, the functional rate $H_G(X)$, and the observed rate using LZW algorithm for different $P_X$.}
\label{fig: lzw}
\end{figure}

\subsection{Randomized quantization and polar coding for computing}\label{subsec: prob_quant_polar}

In this subsection, we consider the case when there is no side information and $\epsilon > 0$. This can be easily extended to the case when $X$ is independent of $Y$ and the main idea of coding remains the same. Note that even in the absence of side information we might not have non-overlapping clustering in the hypergraph, i.e.\ a vertex might belong to more than one hyperedge, hence, quantization followed by universal source coding might not be optimal. This can be observed from Ex.~\ref{example: unique_clustering} with slight modification as well. In this subsection, we show that even when we do not have unique clustering in the hypergraph, we can have practical coding schemes using \emph{randomized quantization} and polar coding. To that end, we provide the following two-step algorithm for attaining the optimal rate asymptotically for any $\epsilon > 0$.
\begin{itemize}[leftmargin=*]
\item \emph{Randomized quantization:} We refer to the process of finding an auxiliary random variable $W$ as randomized quantization since unlike general rate distortion problems, the support set of $W$ is finite and can be determined directly from the corresponding characteristic hypergraph. Hence, every vertex of the hypergraph quantizes to the hyperedges associated with it in a randomized manner. Note that this process is different from random binning in the sense that random binning involves assigning bins to $N$-length sequences for a coding scheme with blocklength $N$, whereas, in randomized quantization we assign probabilities to single elements in $\mathcal{X}$.
Once we have formed the hypergraph, we need to optimize $I(W;X)$ over all conditional probabilities $p(w|x)$ such that $x$ lies in the hyperedge $w$. This is a convex optimization problem over finite variables and can be solved easily. Once we find a suitable $W$, we know from the proof of Thm.~\ref{thm: 2Delta} that we can find a function $g$ such that $\mathbb{E}[d_{\epsilon}(X,g(W))] = 0$, where $g(w)$ is the center of the smallest enclosing circle of the set of points $\{f(x): x \in w \text{ and } p(x)>0\}$. We will assume that $|\mathcal{W}| = 2$ for simplicity, which can be generalized to arbitrary finite-sized $\mathcal{W}$ using ideas from \cite{SasogluTA2009}.

\item \emph{Polar coding:} Once we have found $W$ corresponding to the optimal rate, the next step is to use polar codes to achieve a rate of $I(W;X)$. Define a distortion function $d(x,w) = d_{\epsilon}(x,g(w))$. Then for the chosen $W$, we have $\mathbb{E}[d(X,W)] = 0$. From Thm.~\ref{thm: polar_codes} there exists a set $\mathcal{I} \subset \{1, \ldots, N\}$ and frozen set $\mathcal{I}^c = \{1,\ldots,N\}\backslash \mathcal{I}$ such that $|\mathcal{I}| = NR > NI(W;X)$ and 
\[
Z(U_i|U_1^{i-1},X_1^N) \geq 1 - 2^{-N^{\beta}} \text{  or  } Z(U_i|U_1^{i-1}) \leq 2^{-N^{\beta}},
\]
for $\beta < 1/2$ and $N$ sufficiently large. The coding scheme follows the polar coding scheme for lossy compression \cite{HondaY2013} and is described next.

\emph{Codebook generation:} Let $\mathcal{L}_i$ be the family of functions $\lambda_i : \{0,1\}^{i-1} \to \{0,1\}$ and let $\lambda_{\mathcal{I}^c} \in \prod_{i \in \mathcal{I}^c} \mathcal{L}_i$ be shared between the encoder and the decoder. Later we will show that such set of functions $\lambda_{\mathcal{I}^c}$ exist that give us the desired rate and distortion.\\
\emph{Encoder:} For $i \in \mathcal{I}$, the encoder determines $u_i$ as follows
\[
u_i = \begin{cases}
0 \hspace{2em} \text{with probability} & P_{U_i|U_1^{i-1}}(0|u_1^{i-1},x_1^N)\\
1 \hspace{2em} \text{with probability} & P_{U_i|U_1^{i-1}}(1|u_1^{i-1},x_1^N),\\
\end{cases}
\]
and for $i \in \mathcal{I}^c$, the encoder determines $u_i$ as $u_i = \lambda_i(u_1^{i-1})$. The encoder sends $u_{\mathcal{I}}$ to the decoder. Hence the rate of coding is $|\mathcal{I}|/N$.\\
\emph{Decoder:} The decoder upon receiving $u_{\mathcal{I}}$, determines $u_{\mathcal{I}^c}$ as $u_i = \lambda_i(u_1^{i-1})$ for $i \in \mathcal{I}^c$ and outputs $\hat{w}_1^N = x_1^N G_N$.\\
\emph{Analysis:} We want $\mathbb{E}[d(X,\hat{W})] \to 0$ as $N \to \infty$ for some function $\lambda_{\mathcal{I}^c} \in \prod_{i \in \mathcal{I}^c} \mathcal{L}_i$. For a fixed $N$, the average distortion is given by
$
D_N(\lambda_{\mathcal{I}^c}) = \tfrac{\mathbb{E}[d^N(X^N, \hat{W}^N)]}{N},
$
where $d^N(x^N,w^N) = \sum_{i=1}^N d(x_i, w_i)$. From \cite[Thm.~4]{HondaY2013}, it directly follows that there exists a set of functions $\lambda_{\mathcal{I}^c} = \{\lambda_i \in \mathcal{L}_i\}_{i \in \mathcal{I}^c}$ such that $D_N = \mathcal{O}(2^{-N^{\beta'}})$ for some $\beta' < \beta < 1/2$. Thus, we have $D_N(\lambda_{\mathcal{I}^c}) \to 0$ as $N \to \infty$.
\end{itemize}

\section{Properties of $R(\epsilon)$}\label{sec: properties}

\tikzstyle{block_f_1} = [rectangle, draw,
    text width=8em, text centered, minimum height=2.5 em, node distance = 1.2cm]

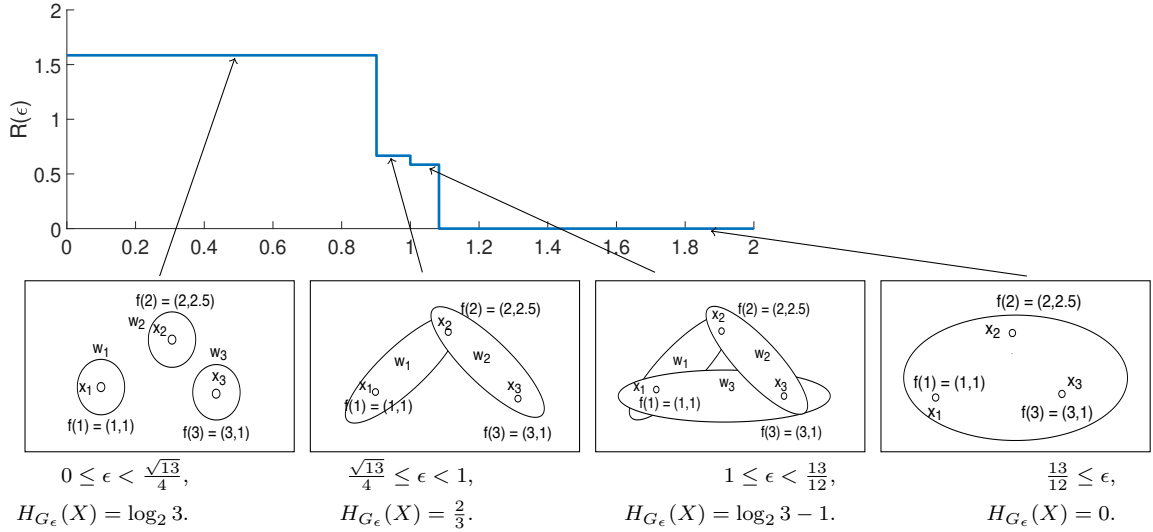
\begin{figure}
\centering
\begin{tikzpicture}
\node(plot){\includegraphics[scale = 0.25]{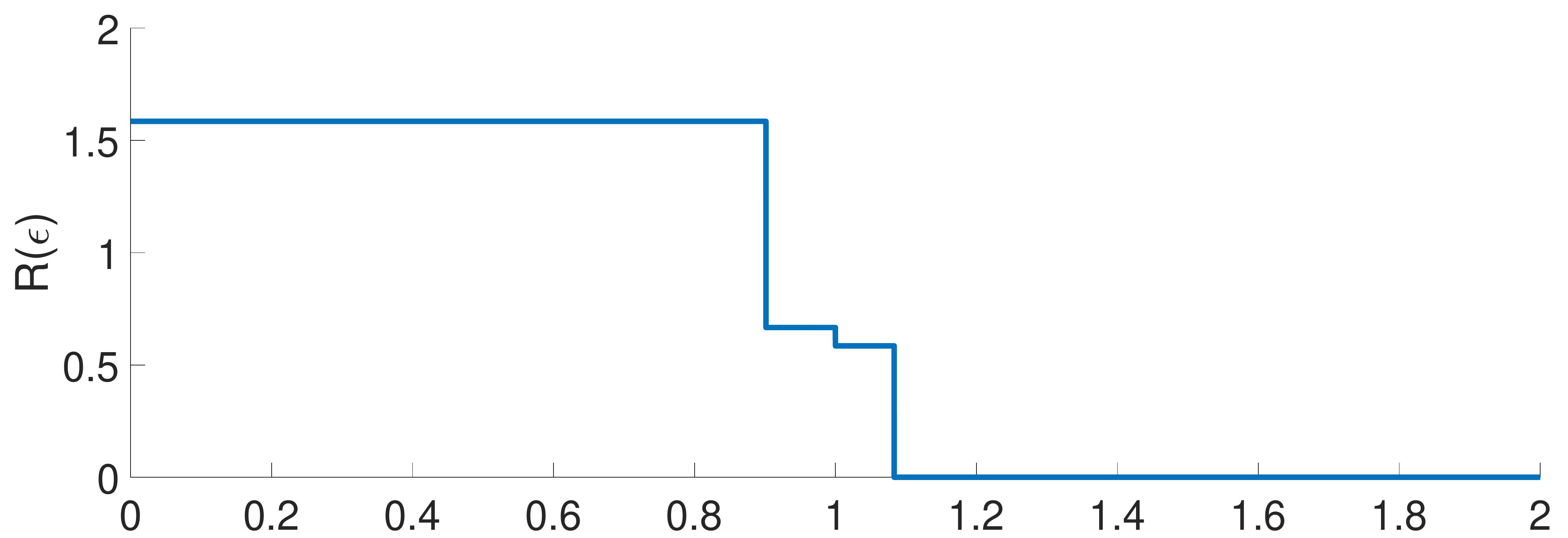}};
\hspace{-3cm}
\node[block_f_1, below = 0.15 cm of plot](h_1){\includegraphics[height = 2cm, width = 3cm]{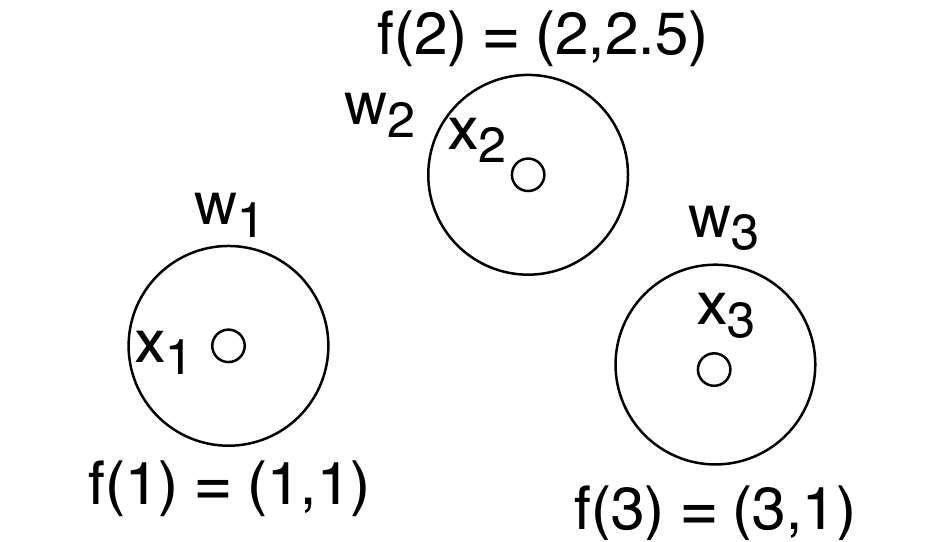}}; \hspace{-1cm}
\node[block_f_1, right = of h_1](h_2){\includegraphics[height = 2cm, width = 3cm]{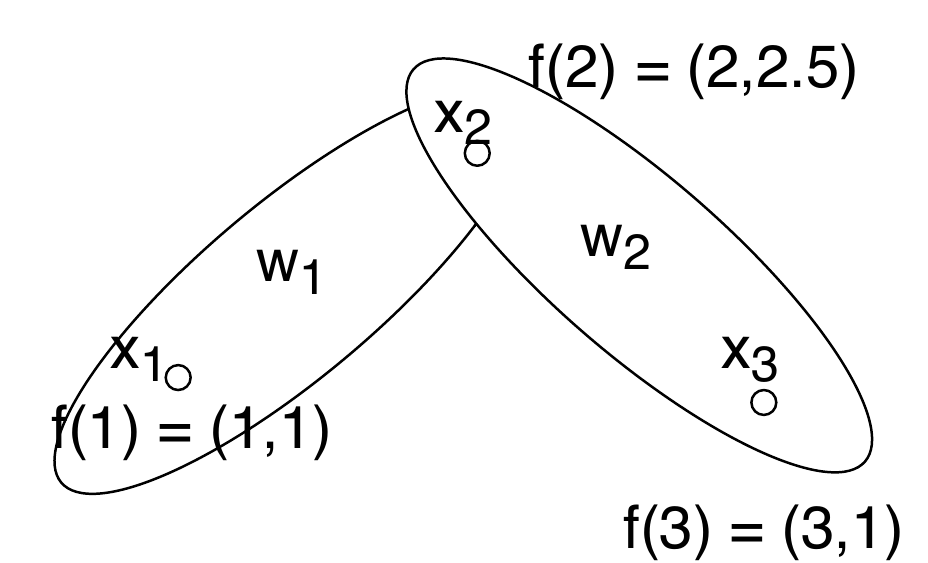}};\hspace{-1cm}
\node[block_f_1, right = of h_2](h_3){\includegraphics[height = 2cm, width = 3cm]{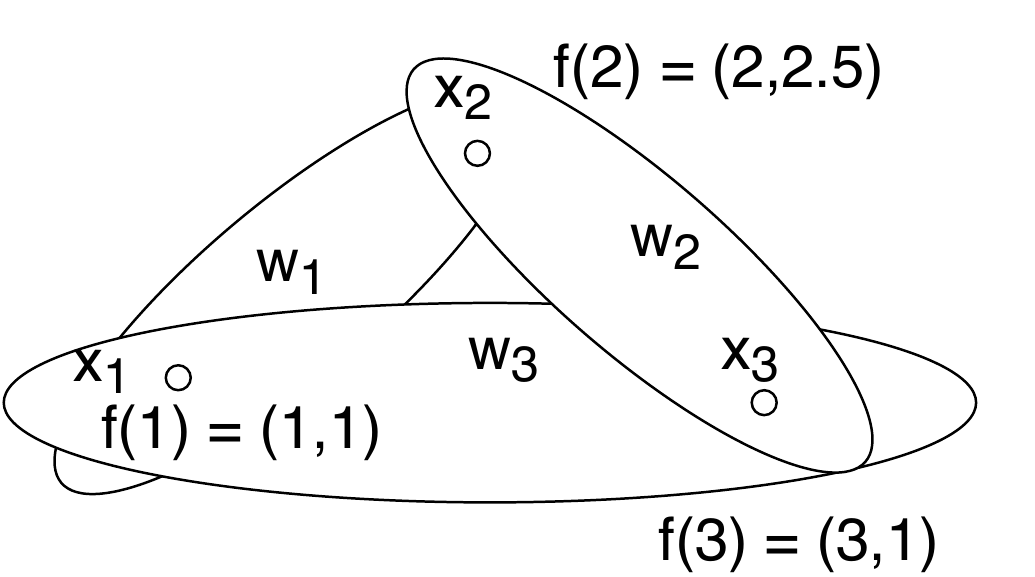}};\hspace{-1cm}
\node[block_f_1, right = of h_3](h_4){\includegraphics[height = 2cm, width = 3cm]{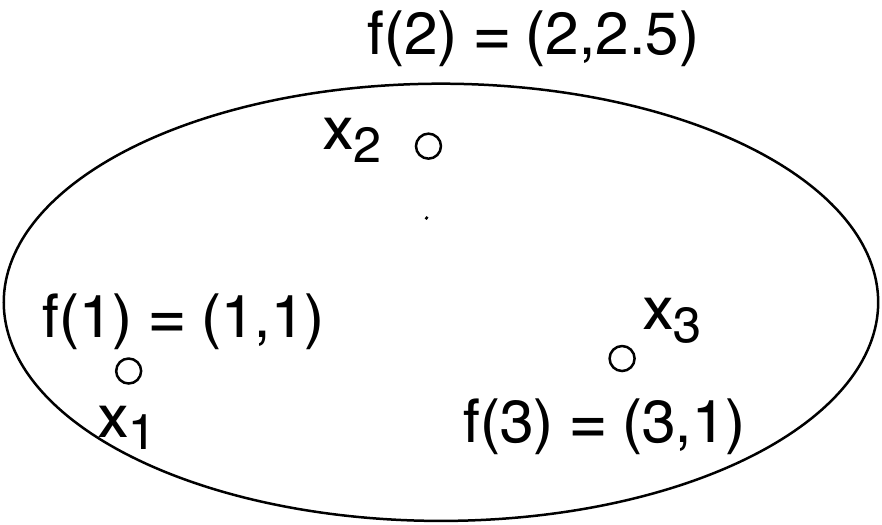}};

\draw[->] (h_1)+(3cm, 1.2cm) -- node[] { }+(4cm,4.1cm);
\draw[->] (h_2)+(1.7cm, 1.25cm) -- node[] { }+(1.28cm,2.77cm);
\draw[->] (h_3)+(0cm, 1.25cm) -- node[] {} +(-3cm,2.63cm);
\draw[->] (h_4)+(0.15cm, 1.2cm)  -- node[] {} +(-4.05cm,1.8cm);
\end{tikzpicture}

\begin{tabular}{r@{\hskip 2 cm}r@{\hskip 2 cm}r@{\hskip 2 cm}r}
\s{ $0 \leq \epsilon < \frac{\sqrt{13}}{4}$,} & \s{$\frac{\sqrt{13}}{4} \leq \epsilon < 1$,} &\s{$1 \leq \epsilon < \frac{13}{12}$,} & \s{$\frac{13}{12} \leq \epsilon$,}\\
\s{$H_{G_{\epsilon}}(X) = \log_2{3}$.} & \s{$H_{G_{\epsilon}}(X) = \frac{2}{3}$.} &\s{$H_{G_{\epsilon}}(X) = \log_2{3} - 1$.} & \s{$H_{G_{\epsilon}}(X) = 0$.}
\end{tabular}
\caption{$R(\epsilon)$ and corresponding hypergraphs.}
\label{fig: discontinuous_rate}
\end{figure}

Clearly $R(\epsilon)$ is a non-increasing function of $\epsilon$. In this section, we show $R(\epsilon)$ may be a discontinuous function of $\epsilon$ and hence from an operational point of view one must design codes with $\epsilon$ close to zero or the right of point of discontinuity for efficient compression algorithms.\cmmnt{Also, note that } We cannot use time-sharing to remove the discontinuity in $R(\epsilon)$ because we have considered maximal distortion\cmmnt{and not expected distortion}. Further, the discontinuity of $R(\epsilon)$ is not obvious from Thm.~\ref{thm: I:coding theorem} but only from the equivalent definition of $R(\epsilon)$ as $H_{G_{\epsilon}}(X)$ that the discontinuity and points of discontinuity can be observed.
We illustrate this property using the following example. 

\begin{example}
Consider the function $f: \mathcal{X} \to \mathcal{Z}$, where $\mathcal{Z} \subset \mathbb{R}^2$, and $f$ is defined as $f(1) = (1,1)$, $f(2) = (2,2.5)$, $f(3) = (3,1)$. Then for different values of $\epsilon$ we have different $G_{\epsilon}$ as illustrated in Fig.~\ref{fig: discontinuous_rate}. $R(\epsilon)$ depends on $G_{\epsilon}$ and hence, on increasing the value of $\epsilon$, the values of $\epsilon$ where $G_{\epsilon}$ changes are the points of discontinuity of $R(\epsilon)$ as illustrated in Fig.~\ref{fig: discontinuous_rate}. 
\label{example: discontinuous_rate}
\end{example}
\cmmnt{
\begin{figure}
\centering
\begin{subfigure}{.45\textwidth}
  \centering
  \includegraphics[width=0.8\linewidth]{u_discontinuous_rate_1}
  \caption{$0 \leq \epsilon < \frac{\sqrt{13}}{4}$, $H_{G_{\epsilon}}(X) = \log_2{3}$.}
  \label{fig:sub1}
\end{subfigure}%
\begin{subfigure}{.45\textwidth}
  \centering
  \includegraphics[width=0.8\linewidth]{u_discontinuous_rate_2}
  \caption{$\frac{\sqrt{13}}{4} \leq \epsilon < 1$, $H_{G_{\epsilon}}(X) = \frac{2}{3}$.}
  \label{fig:sub2}
\end{subfigure}

\begin{subfigure}{.45\textwidth}
  \centering
  \includegraphics[width=0.8\linewidth]{u_discontinuous_rate_3}
  \caption{$1 \leq \epsilon < \frac{13}{12}$, $H_{G_{\epsilon}}(X) = \log_2{3} - 1$.}
  \label{fig:sub3}
\end{subfigure}
\begin{subfigure}{.45\textwidth}
  \centering
  \includegraphics[width=0.65\linewidth]{u_discontinuous_rate_4}
  \caption{$\frac{13}{12} \leq \epsilon$, $H_{G_{\epsilon}}(X) = 0$.}
  \label{fig:sub4}
\end{subfigure}
\caption{The maximal hyperedges formed and the hypergraph entropy corresponding to the characteristic graph $G_{\epsilon}$ for different values of $\epsilon$ in Ex.~\ref{example: discontinuous_rate}.}
\label{fig:Delta_clustering}
\end{figure}
}

\section{Conclusion}\label{sec: conclusion}
This paper considers the problem of coding for computing with a fidelity constraint. The main insight regarding the solution of the problem is obtained by characterizing the rate as the conditional entropy of a hypergraph, which we call functional $\epsilon$-entropy. It is shown that the rate-distortion function for the problem is discontinuous with respect to the fidelity constraint. We also develop practical coding schemes for the problem and provide achievable bounds when the exact function is unknown to the encoder but an approximate function or a class to which the function belongs is known.
The rate provided in this paper for a maximal distortion $\epsilon$ can be seen as an upper-bound to the rate for the rate distortion problem with expected distortion $\epsilon$.
For future work, we want to provide stronger practically achievable bounds for the problem of coding for computing with expected distortion, since practical codes for this problem are still unknown.
\section*{Acknowledgement}\label{sec: acknowledgement}
We appreciate valuable discussions with Souktik Roy, Harshit Yadav, Aditya Deshmukh, Akshayaa Magesh, and Ishita Jain.
\Section{References}
\bibliographystyle{IEEEtran}
\bibliography{abrv,conf_abrv,lrv_lib}
\end{document}